\documentclass[11pt]{article}
\usepackage[utf8]{inputenc}
\usepackage{fullpage}
\usepackage{amssymb}

\usepackage[symbol]{footmisc}

\usepackage{amsmath,bbm}
\usepackage{amsfonts,amssymb}

\usepackage{hyperref}

\usepackage{graphicx}
\usepackage{epsfig}
\usepackage{times}
\usepackage{verbatim}
\usepackage{color}

\newcommand{\be}{\begin{equation}}
\newcommand{\ee}{\end{equation}}
\newcommand{\bq}{\begin{eqnarray}}
\newcommand{\eq}{\end{eqnarray}}
\newcommand{\bea}{\begin{eqnarray}}
\newcommand{\eea}{\end{eqnarray}}
\newcommand{\ba}{\begin{align}}
\newcommand{\ea}{\end{align}}

\newcommand{\Tr}{{\rm Tr}}

\newcommand{\bC}{\mathbbm{C}}

\newcommand{\cV}{\mathcal{V}}

\newcommand{\cL}{\mathcal{L}}
\newcommand{\cA}{\mathcal{A}}

\newcommand{\cW}{\mathcal W}
\newcommand{\cM}{\mathcal M}

\newcommand{\cK}{\mathcal K}

\newcommand{\cT}{\mathcal T}
\newcommand{\cQ}{\mathcal Q}

\newcommand{\calS}{\mathcal{S}}

\newtheorem{theorem}{Theorem}
\newtheorem{lemma}[theorem]{Lemma}

\newtheorem{definition}[theorem]{Definition}

\def\qed{\leavevmode\unskip\penalty9999 \hbox{}\nobreak\hfill
    \quad\hbox{\leavevmode  \hbox to.77778em{%
               \hfil\vrule   \vbox to.675em%
               {\hrule width.6em\vfil\hrule}\vrule\hfil}}
     \par\vskip3pt}
\newenvironment{remark}{\vspace{1.5ex}\par\noindent{\it Remark}}
    {\hspace*{\fill}$\Box$\vspace{1.5ex}\par}

\definecolor{mygray}{gray}{0.6}

\renewcommand{\>}{\rangle}
\newcommand{\<}{\langle}

\newcommand{\cJ}{\mathcal{J}}

\usepackage{tikz}
\usetikzlibrary{quantikz}
\usepackage{adjustbox}

\title{Szegedy Walk Unitaries for Quantum Maps}

\author{
Pawel Wocjan$^*$ \and Kristan Temme\thanks{IBM Quantum, IBM T.J. Watson Research Center, 
Yorktown Heights, NY 10598, USA. 
\newline
\textsf{\{Pawel.Wocjan, kptemme\}@ibm.com}} 
}

\begin{document}

\maketitle


\abstract{Szegedy developed a generic method for quantizing classical algorithms based on random walks [Proceedings of FOCS, 2004, pp. 32-41]. A major contribution of his work was the construction of a walk unitary for any reversible random walk. Such unitary posses two crucial properties: its eigenvector with eigenphase $0$ is a quantum sample of the limiting distribution of the random walk and its eigenphase gap is quadratically larger than the spectral gap of the random walk.

It was an open question if it is possible to generalize Szegedy's quantization method for stochastic maps to quantum maps. We answer this in the affirmative by presenting an explicit construction of a Szegedy walk unitary for detailed balanced Lindbladians -- generators of quantum Markov semigroups -- and detailed balanced quantum channels. 

We prove that our Szegedy walk unitary has a purification of the fixed point of the Lindbladian as eigenvector with eigenphase $0$ and that its eigenphase gap is quadratically larger than the spectral gap of the Lindbladian.  To construct the walk unitary we leverage a canonical form for detailed balanced Lindbladians showing that they are structurally related to Davies generators.  We also explain how the quantization method for Lindbladians can be applied to quantum channels.

We give an efficient quantum algorithm for quantizing Davies generators that describe many important open-system dynamics, for instance, the relaxation of a quantum system coupled to a bath. Our algorithm extends known techniques for simulating quantum systems on a quantum computer. }

\newpage
\tableofcontents


\section{Introduction}


Szegedy developed a generic method for quantizing algorithms based on random walks\footnote{We will use the terms random walk, Markov chain, stochastic matrix, and stochastic map synonymously throughout the manuscript.} \cite{szegedy04}. 
An important contribution of this work was the quantization of any reversible (also called detailed balanced) random walk in the following sense: Let $P=(p_{yx})_{x,y\in\Omega}$ denote the (column) stochastic matrix representing the reversible random walk on state space $\Omega$ with limiting distribution $\pi=(\pi_x)_{x\in\Omega}$; Szegedy showed how to construct a corresponding walk unitary $W(P)$ so that its unique eigenvector with eigenvalue $1$ (or equivalently with eigenphase $0$) is the quantum sample\footnote{More precisely, the walk unitary $W(P)$ acts on the Hilbert space $\bC^{|\Omega|}\otimes \bC^{|\Omega|}$ and the corresponding eigenvector with eigenvalue $1$ is the state $|\pi\> = \sum_{x\in\Omega} \sqrt{\pi_x} |x\> \otimes \sum_{y\in\Omega} \sqrt{p_{yx}} |y\>$.} 
\begin{align}\label{eq:qsample_pi}
    |\pi\> = \sum_{x\in\Omega} \sqrt{\pi_x} |x\>
\end{align}
that is a coherent encoding of $\pi$.
The reason for the many quantum speed-ups is that the phase gap of walk unitary $W(P)$ is $\sqrt{\Delta}$, where $\Delta$ denotes the spectral gap of the stochastic matrix $P$. Szegedy's construction can now be understood using the quantum singular value transformation \cite{gilyen19quant_sing_val} that provides a unifying approach to many quantum algorithms and methods.

The purpose of the present work is to study the quantization of quantum maps. Let $\cT$ be any detailed balanced quantum map with fixed point $\sigma$. 

We construct a unitary $W(\cT)$ such that its eigenvector with eigenphase $0$ is the purification 
\begin{align}
    |\sigma^{1/2}\> &= (\sigma^{1/2} \otimes I)|\Omega\>
\end{align}
of $\sigma$, where $|\Omega\>=\sum_{x} |x\> \otimes |x\>$ denotes the (unnormalized) maximally entangled state. We prove that the eigenphase of $W(\cT)$ is quadratically amplified compared to the spectral gap of the quantum map $\cT$.  We also present quantum circuits that efficiently implement $W(\cT)$.  To accomplish these goals we proceed as follows.

First, we develop a generic framework for quantizing continuous time  purely irreversible detailed balanced quantum maps. More precisely, we show how to quantize quantum Markov semigroups. Let $\cL$ denote the Lindbladian, that is, the generator of a quantum Markov semigroup with fixed point $\sigma$. Assuming that $\cL$ satisfies the detailed balanced condition with respect to $\sigma$, we show how to construct a unitary $W(\cL)$ such that its eigenvector with eigenphase $0$ is a purification of $\sigma$ having the above form.
We leverage the fact that detailed balanced Lindbladians can be expressed in a certain canonical form showing that they essentially have the structure of Davies generators. 
We relate the unitary $W(\cL)$ to a quantum discriminate $\cQ$, which arises from $\cL$ through a similarity transformation defined with respect to the fixed point $\sigma$. Analogously to the classical setting the quantum discriminate $\cQ$ is shown to be independent of the fixed point $\sigma$ and even the particular form in which the detailed balance condition is stated (due to the non-commutative nature of quantum mechanics, there are several natural notions of detailed balanced).  Centrally, we prove that the phase gap of $W(\cL)$ is quadratically amplified compared to the spectral gap of $\cL$.   We also discuss how the problem of quantizing detailed balanced quantum channels reduces to the problem of quantizing detailed balanced Lindbladians.

Second, we examine for which quantum maps the above quantization method can be realized efficiently on a quantum computer. We focus on the special case of Davies generators that have thermal Gibbs states $e^{-\beta H} / \Tr\big({e^{-\beta H} \big)}$ for Hamiltonian $H$ and inverse temperature $\beta$ as fixed points. We show how to efficiently quantize  Davies generators provided that the energies of the Hamiltonian $H$ can be resolved, for instance, when we have access to a block-encoding of the Hamiltonian $H$ whose energies satisfy a rounding promise and we use the energy estimation method in \cite{rall2021}.

The quadratic gap amplification is at the heart the quantum speed-ups of many  classical random walk based algorithms \cite{somma2008quantum, poulin2009, harrow2020adaptive, arunachalam2020simpler}. This polynomial speed-up is measured relative to the mixing time of the classical walk. For quantum maps, recent work \cite{olkiewicz1999hypercontractivity,Chi-squared-mising-stuff,kastoryano2013quantum,bardet2017estimating,bardet2018hypercontractivity,datta2020relating} has investigated the mixing behavior of the corresponding Markov process. Tools have been developed to bound the mixing time and to analyze the spectral gaps of detailed balanced quantum maps. For several explicit examples \cite{temme2013lower,kastoryano2016quantum, temme2017thermalization, capel2020modified,bardet2021group,bardet2021modified,lucia} spectral gap bounds could be obtained.

A special case of a detailed balanced quantum map was introduced in the context of the quantum Metropolis algorithm \cite{temme2011quantum} to prepare the Gibbs state of a quantum Hamiltonian with a time evolution that can be simulated efficiently. The Gibbs state can be prepared efficiently if the corresponding quantum map is rapidly mixing. A Szegedy walk unitary has been constructed for the classical Metropolis-Hastings algorithm \cite{somma2008quantum}. The Metropolis-Hastings random walk prepares the Gibbs state for a Hamiltonian that is diagonal in the computational basis. An extension of this walk algorithm to non-diagonal quantum Hamiltonian subject to specific assumptions has been constructed in \cite{yung}. This work assumes access to a projective measurement that enables one to distinguish between some arbitrary but fixed eigenvector basis of the system Hamiltonian.  This assumption essentially makes it possible to reduce the problem to a classical Metropolis random walk on the fixed eigenbasis. We emphasize that when the system Hamiltonian has degenerate spectra such eigenbasis measurement cannot be realized by measuring the energies of the Hamiltonian -- even if one can perfectly resolve the energies.  This is because one cannot distinguish between some arbitrary vectors within a degenerate eigenspace without any additional assumptions. Furthermore, in the generic situation of energy estimation with finite resources, any quantum Hamiltonian on an exponentially large Hilbert space with polynomial bounded operator norm will exhibit  degeneracies. 
The present work enables the direct ``quantization''  of thermalizing quantum maps such as Davies generators and therefore does not need to make any assumptions on the identifiabilty of some eigenbases.

The paper is organized as follows.
In Section~\ref{sec:Lindbladians}, we give an introduction to detailed balanced Lindbladians and present a canonical form for them.  We also state our main result on quantizing detailed balanced Lindbladian. Its proof is presented  in the last section, leveraging the results for Davies generators that are worked out in Sections~\ref{sec:Davies} and 
\ref{sec:Davies_approximate}.
In Section~\ref{sec:Davies}, we show how to quantize Davies generators assuming that the coupling operators are reflections and that ideal energy estimation of the Hamiltonian is possible.
In Section~\ref{sec:Davies_approximate}, we present an efficient quantum algorithm for quantizing Davies generators using the algorithm for approximate energy estimation in \cite{rall2021}.
In Section~\ref{sec:generalLindbladians}, we consider again general detailed balanced Lindbladians, that is, the coupling operators are not restricted to be reflections. Starting from a canonical form for general detailed balanced Lindbladians, we show that they can be always represented as Davies generators having reflections as coupling operators with the help of suitable block encodings. We conclude with open research problems.


\section{Canonical form for detailed balanced quantum maps}\label{sec:Lindbladians}

We begin by describing detailed balanced Lindbladians and presenting a canonical form for them.  We later explain how any detailed balanced quantum channel can be related to a detailed balanced Lindbladian. 

\subsection{Detailed balanced Lindbladians}

Let $\cL$ denote the generator of a quantum Markov semigroup. We consider purely irreversible Lindbladians, that is, their Hamiltonian part is equal to zero. 
We choose the Schr\"odinger picture so the Lindbladian $\cL$ describes the evolution of quantum states $\rho$ according to
\begin{align}\label{eq:masterEq}
    \dot\rho = \cL(\rho)
\end{align}
Any (purely irreversible) Lindbladian $\cL$ can be written in the form
\begin{align}
    \cL(\rho) = \sum_i 
        L_i \rho L_i^\dagger - 
        \tfrac{1}{2} \{ L_i^\dagger L_i,  \rho \}
\end{align}
where $L_i$ are the so-called Lindblad operators or jump operators. 
We require for the remainder of the manuscript that the operator norm of $\cL$ satisfies $\|\cL\| \leq 1$.  This can always be achieved by by renormalizing.  

The Lindbladian $\cL^*$ in the Heisenberg picture has the form
\begin{align}
    \cL^*(O) = \sum_i 
        L_i^\dagger O L_i - 
        \tfrac{1}{2} \{ L_i^\dagger L_i, O \}
\end{align}
where $O$ denotes an arbitrary observable. 
We point out that the Lindbladian $\cL^*$ in the Heisenberg picture is the adjoint of the Lindbladian $\cL$ in the Schr\"odinger picture with respect to the Hilbert Schmidt (HS) inner product, which is defined by $\<X,Y\>=\Tr(X^\dagger Y)$ for all matrices $X,Y$. This elementary observation will become important later. 

There are no restrictions on the jump operators $L_i$ for a general Lindbladian. 
However, it turns out that the jump operators of \emph{detailed balanced} Lindbladians have very special and useful properties that enable us to construct a corresponding walk unitary with all the desired properties.  We proceed by defining detailed balanced Lindbladians and discuss a canonical form that is the starting point for our quantization method.  

In the manuscript \cite{Chi-squared-mising-stuff} a notion of quantum detailed balance condition (DBC) was defined for every operator monotone function that obeys rather simple symmetry conditions. It was argued that not all these DBC notions are equivalent. However, it was recently shown \cite{Carlen2020} that a stronger detailed balance condition referred to as $\sigma$-DBC actually makes almost all the conditions collapse to a single condition and also enforces that the $\sigma$-DBC map always has a particular form. 

Before introducing the quantum setting, we review the classical setting to explain the analogies and differences between the two.
If $P=(P_{yx})$ is the transition matrix of a Markov chain on finite state space $\Omega$ with limiting distribution (probability vector) $\pi$, we say that the detailed balanced condition holds if $P_{yx}\pi_x = P_{xy} \pi_y$. An analytic way to formulate this condition is that $P$ is hermitian with respect to the weighted inner product on $\bC^{|\Omega|}$ defined by
$\<v, w\>_1=\sum_{x\in\Omega} \pi_x \bar{v}_x w_x$. Note that we use the subscript $1$ to denote the weighted inner product to distinguish it from the standard inner product, but also to explain the analogy to the GNS inner product in the quantum setting. Hermiticity of $P$ with respect to this inner product implies that $\<Pv, w\>_1=\<v,Pw\>_1$ holds for all vectors $v,w\in\bC^{|\Omega|}$. Setting $v$ and $w$ to be the two standard basis vectors $|x\>$ and $|y\>$, allows us to recover the original definition of the detailed balanced condition.  We can also state this condition in matrix form as $D P^T=P D$, where $D$ is the diagonal matrix whose entries are the $\pi_x$ entries of the limiting distribution.  Equivalently, we have $D^{1/2} P^T D^{-1/2}=D^{-1/2} P D^{1/2}$.  Observe that the RHS of this equality is the discriminate matrix of the reversible Markov chain. We see that this discriminate matrix is hermitian with respect to the standard inner product.

In the quantum setting, with reference matrix $\sigma$ of full rank that is not a multiple of the identity matrix, there are many candidates for such a weighted inner product. Let
\begin{align}
    \sigma = \sum_{k=1}^m p_k P_k
\end{align}
be the spectral decomposition of $\sigma$, where 
$p_k\in (0,1)$ are the $m\ge 2$ different eigenvalues with $\sum_{k=1}^m p_k=1$, and the orthogonal projectors form a resolution of the identity $\sum_{k=1}^m P_k = I$. 

For instance, given $\sigma$ as above and $s\in [0,1]$ one can define a weighted inner product by
\begin{align}\label{eq:inner_s}
    \<A, B\>_s = \<A, \sigma^{1-s} B \sigma^s \>
\end{align}
where $\<X, Y\>=\Tr(X^\dagger Y)$ denotes the Hilbert-Schmidt (HS) inner product.  There are two important special cases. For $s=1$, the inner product is called the Gelfand-Naimark-Segal (GNS) inner product, and for $s=1/2$ it is called the Kubo-Martin-Schwinger (KMS) inner product. 

The weighted inner product in (\ref{eq:inner_s}) can be written as
\begin{align}\label{eq:inner_s2}
    \<A, B\>_s = \<A, \Delta_\sigma^{1-s} B \sigma\>
\end{align}
where $\Delta_\sigma$ denotes the modular operator defined by $\Delta_\sigma X=\sigma X \sigma^{-1}$.
More generally, given any function $f:(0,\infty) \rightarrow (0,\infty)$, one can define a weighted inner product
\begin{align}\label{eq:inner_f}
    \<A, B\>_f = \<A, f(\Delta_\sigma) B \sigma\>
\end{align}
Above we use the convention that the superoperator $f(\Delta_\sigma)$ acts on $B$ alone and not on the product of $B$ and $\sigma$.
Note that $\<\cdot,\cdot\>_1$ is the GNS inner product whether $1$ is interpreted as a number, as in (\ref{eq:inner_s}), or as the constant function $f(t)=1$, as in (\ref{eq:inner_f}).

Define the superoperator $\Omega_\sigma^f=R_\sigma\circ f(\Delta_\sigma)$, where $R_X$ denotes right multiplication by $X$.  Then another way to write the weighted inner product in (\ref{eq:inner_f}) is 
\begin{align}\label{eq:inner_f_with_Omega}
    \<A,B\>_f = \<A, \Omega_\sigma^f B\>
\end{align} 

It will be useful to establish some properties of $\Omega_\sigma^f$ in the lemma below.

\begin{lemma}[Properties of $\Omega_\sigma^f$]\label{lem:Omega_sigma_f}
We have
\begin{itemize}
\item[(i)] $\Omega_\sigma^f$ acts as follows
\begin{align}
    \Omega_\sigma^f A = \sum_{k,\ell} f(p_k p_\ell^{-1}) p_\ell P_k A P_\ell
\end{align}
\item[(ii)] $[\Omega_\sigma^f]^{-1}$ is $R_{\sigma^{-1}} \circ (1/f)(\Delta_\sigma)$, where the function $(1/f)$ is defined by $(1/f)(x)=1/f(x)$ for all $x\in(0,\infty)$, and acts as follows
\begin{align}
    [\Omega_\sigma^f]^{-1} A = \sum_{k,\ell} (1/f)(p_k p_\ell^{-1}) p_\ell^{-1} P_k A P_\ell
\end{align}
\item[(iii)] $\Omega_\sigma^f$ is positive definite with respect to the HS inner product.
\end{itemize}
\end{lemma}

\begin{proof}
Recall that the eigenprojectors $P_k$ of $\sigma$ form a resolution of the identity so we can write $A=\sum_{k,\ell} P_k A P_\ell$.  We have $\Delta_\sigma(P_k A P_\ell)=p_k p_\ell^{-1} P_k A P_\ell$. Thus, $f(\Delta_\sigma)(P_k A P_\ell) = f(p_k p_\ell^{-1}) P_k A P_\ell$.  We have $R_\sigma (P_k A P_\ell) = p_\ell P_k A P_\ell$.  The two superoperators $f(\Delta_\sigma)$ and $R_\sigma$ commute. These observations establish the first statement. The second statement follows from $(1/f)(\Delta_\sigma) \circ f(\Delta_\sigma) = I$ and $R_{\sigma}^{-1} \circ R_\sigma = I$.
The third statement follows from 
\begin{align}
    \< A, \Omega_\sigma^f A\> = \sum_{k,\ell}
    f(p_k p_\ell^{-1}) \cdot p_\ell \cdot \< P_k A P_\ell, P_k A P_\ell\> > 0
\end{align}
for all non-zero $A$. \qed{}
\end{proof}

Since $\Omega_\sigma^f$ is hermitian with respect to HS inner product, we can equivalently apply $\Omega_\sigma^f$ to the left argument of the HS inner product in (\ref{eq:inner_f_with_Omega}).

Let $\cM$ be an arbitrary linear operator and $\cM^*$ its adjoint with respect to the HS inner product. Then, we have
\begin{align}
    \< A, \cM B\>_f 
    &= 
    \< \Omega_\sigma^f A, \cM B\> \\
    &=
    \< \cM^*(\Omega_\sigma^f A), B\> \\
    &=
    \big\<  [\Omega_\sigma^f]^{-1} \big(\cM^*(\Omega_\sigma^f A)\big), B\big\>_f
\end{align}
This means that $\cM$ is hermitian with respect to the weighted inner product $\<\cdot,\cdot\>_f$ if and only if
\begin{align}\label{eq:hermitian}
    \Omega_\sigma^f \circ \cM = \cM^* \circ \Omega_\sigma^f
\end{align}

It was shown in \cite{Carlen2020} that if a linear operator $\cM$ is hermitian with respect to $\<\cdot,\cdot\>_s$ for \emph{some} $s\in I=[0,\frac{1}{2})\cup (\frac{1}{2},1]$, then it is automatically hermitian with respect $\<\cdot,\cdot\>_s$ for \emph{all} $s\in[0,1]$, including $s=\frac{1}{2}$. Moreover, it is also hermitian with respect to $\<\cdot,\cdot\>_f$ for any $f : (0,\infty)\rightarrow (0,\infty)$.

\begin{definition}[Detailed balanced]
The Lindbladian $\cL$ in the Schr\"odinger picture is detailed balanced with respect to the state $\sigma$ if and only if the Lindbladian $\cL^*$ in the Heisenberg picture is hermitian with respect to the $\sigma$-GNS inner product.  
\end{definition}

\begin{lemma}[Detailed balanced implies fixed point]\label{lem:DB_implies_fixed_point}
Let $\cL$ be detailed balanced with respect to the state $\sigma$. Then, $\sigma$ is a fixed point of $\cL$, that is, $\cL(\sigma)=0$.
\end{lemma}

\begin{proof}
Using that $\cL^*(I)=0$ and that $\cL^*$ is hermitian with respect to the $\sigma$-GNS inner product, we obtain 
\begin{align}
    0
    =
    \< \cL^*(I), A\>_1 
    =
    \< I, \cL^*(A)\>_1 
    =
    \Tr\big(\sigma \cL^*(A)\big) 
    =
    \Tr\big(\cL(\sigma) A\big) 
\end{align}
for an arbitrary observable $A$. The condition $\Tr(\cL(\sigma)A)=0$ for all $A$ implies that $\cL(\sigma)=0$. \qed{}
\end{proof}

If $\cL$ is detailed balanced with respect to $\sigma$, then it holds that
\begin{align}\label{eq:L_db_implication}
    \Omega_\sigma^f \circ \cL^* = \cL \circ \Omega_\sigma^f 
\end{align}
This follows from the characterization in (\ref{eq:hermitian}) with $\cM=\cL^*$ and $\cL^{**}=\cL$. The condition in (\ref{eq:L_db_implication}) can also be expressed as
\begin{align}\label{eq:L_db_implication_2}
    [\Omega_\sigma^f]^{1/2} \circ \cL^* \circ [\Omega_\sigma^f]^{-1/2} = 
    [\Omega_\sigma^f]^{-1/2} \circ \cL \circ [\Omega_\sigma^f]^{1/2} 
\end{align}


We are now ready to define the quantum discriminate. 

\begin{definition}[Quantum discriminate]
Let $\cL$ be a detailed balanced Lindbladian with respect to the state $\sigma$. We define the quantum discriminate $\cQ^f$ to be
\begin{align}
    \cQ^f = I + \cK^f
\end{align}
where $\cK^f$ is the operator defined to be the LHS or equivalently the RHS of (\ref{eq:L_db_implication_2}), that is,
\begin{align}
    \cK^f 
    =
    [\Omega_\sigma^f]^{1/2} \circ \cL^* \circ [\Omega_\sigma^f]^{-1/2}
    =
    [\Omega_\sigma^f]^{-1/2} \circ \cL \circ [\Omega_\sigma^f]^{1/2}
\end{align}
\end{definition}


\begin{lemma}[Quantum discriminate is hermitian]
Assume that $\cL$ is detailed balanced with respect to the state $\sigma$. Let $f$ be any function $f:(0,\infty)\rightarrow (0,\infty)$.  Then, we have: (i) the quantum discriminate $\cQ^f$ is hermitian with respect to the Hilbert-Schmidt inner product; (ii) the fixed point of $\cQ^f$ is $\tau=\sigma^{1/2}$, that is, $\cQ^f(\tau)=\tau$; (iii) the spectral gap of $\cL$ and the spectral gap of the quantum discriminate $\cQ^f$ are equal.
\end{lemma}


\begin{proof}
The superoperator $\Omega_\sigma^f$ is hermitian with respect to the HS inner product.  Thus, $[\Omega_\sigma^f]^{\pm 1/2}$ are also hermitian.  We now see that both sides of (\ref{eq:L_db_implication_2}) are adjoints of each other (observe that the adjoint operation reverses the order of the superoperators) and, thus, $\cK^f$ is hermitian. 

Using property (i) in Lemma~\ref{lem:Omega_sigma_f}, we obtain that $[\Omega_\sigma^f]^{1/2}(\tau)$ is proportional to $\sigma$ so that $\cK^f(\tau)=0$.
\qed{}
\end{proof}

We now see that the detailed balanced condition ensures that the Lindbladian $\cL$ has real spectrum since it is similar to the hermitian operator $\cK^f$. Moreover, its spectrum of $\cL$ is contained in the interval $[-1, 0]$; we have $-1$ as the left end of the interval because we assume that the operator norm of $\cL$ satisfies $\|\cL\| \le 1$. The fixed point $\sigma$ of $\cL$ is the eigenvector with eigenvalue $0$.  

To analyze the present case of a detailed balanced Lindbladian $\cL$ with respect to $\sigma$ and later the more special case of a Davies generator, whose fixed point is always a thermal Gibbs state, in a unifying approach, we define the corresponding ``Hamiltonian'' $h$ for $\sigma$ by setting $h=-\ln\sigma$. Its spectral decomposition is 
\begin{align}
    h = \sum_{k=1}^m \varepsilon_k \Pi_k 
\end{align}
where the ``energies'' are simply given by $\varepsilon_k = -\ln p_k$. We can now formally express the state $\sigma$ as a Gibbs state for the Hamiltonian $h$ at inverse temperature $\beta=1$, that is,
\begin{align}
    \sigma = e^{-h}
\end{align}
Note that the corresponding partition function is $Z=\sum_k e^{-\varepsilon_k}=\sum_k p_k=1$. We call the differences of the form $\omega_{k\ell}=\varepsilon_k-\varepsilon_\ell$ the Bohr frequencies of $h$, where $k,\ell=1,\ldots,m$.

It is of paramount importance that a detailed balanced Lindbladian $\cL$ can always be written in the  canonical form presented below,
which shows that detailed balanced Lindbladians essentially have the structure of Davies generators.  This follows directly from the results in \cite[Theorem 2]{alicki76detailed-balanced}, \cite[Remark  on page 101]{kossakowski77qdb}, and \cite[Theorem 2.4]{Carlen2020}. 

\begin{definition}[Canonical form for detailed balanced Lindbladians]\label{def:Lindbladian_canonical}
Let $\cL$ be a (purely irreversible) Lindbladian that is detailed balanced with respect to $\sigma$. We say the Lindbladian $\cL$ is in canonical form if it is given as a convex combination
\begin{align}
    \cL(\rho)=\sum_\alpha w_\alpha \cL^\alpha(\rho)
\end{align}
and each term $\cL^\alpha$ has the following form
\begin{align}\label{eq:cL_alpha}
    \cL^\alpha(\rho) &= 
    \sum_\omega 
    G^\alpha(\omega) \left(
    X^\alpha(\omega) \rho X^{\alpha\dagger}(\omega)
    - \frac{1}{2} \big\{ X^{\alpha\dagger}(\omega) X^\alpha(\omega), \rho \big\}
    \right)
\end{align}
For each $\alpha$, the summation index $\omega$ in (\ref{eq:cL_alpha}) runs over all different Bohr frequencies of the Hamiltonian $h$. The operators $X^\alpha(\omega)$ and scalars $G^\alpha(\omega)$ satisfy the following conditions:
\begin{align}
    X^\alpha(\omega)                &=     X^{\alpha\dagger}(-\omega) \\
    \Delta_\sigma X^\alpha(\omega)  &=     e^{-\omega} X^{\alpha}(\omega) \label{eq:eig_modular} \\
    G^\alpha(\omega)                &\ge 0 \\
    G^\alpha(\omega)                &= e^\omega G^\alpha(-\omega) \label{eq:sqrtGs}
\end{align}
\end{definition}


We now establish some elementary results that we will use to show that the quantum discriminate is independent of $f$. For each Bohr frequency $\omega$, define the set
\begin{align}\label{eq:J_omega}
    \cJ_\omega &= \{ (k,\ell) : \varepsilon_k - \varepsilon_\ell = \omega \}
\end{align}
$B$ is an eigenvector of the modular operator $\Delta_\sigma$ with eigenvalue $e^{-\omega}$ iff 
\begin{align}
    B = \sum_{(k,\ell) \in \cJ_\omega} P_k B P_\ell
\end{align}
This is because the RHS of the equation above is precisely the projection of $B$ onto the eigensubspace of $\Delta_\sigma$ with eigenvalue $e^{-\omega}$.  We will apply this characterization to $X^\alpha(\omega)$, $X^{\alpha\dagger}(\omega)$, and $X^{\alpha\dagger}(\omega) X^\alpha(\omega)$. Recall that $X^\alpha(\omega)$ is an eigenvector of $\Delta_\sigma$ with eigenvalue $e^{-\omega}$ and that $X^{\alpha\dagger}(\omega)=X^{\alpha}(-\omega)$. It is elementary to verify that $X^{\alpha\dagger}(\omega) X^\alpha(\omega)$ is an eigenvector of $\Delta_\sigma$ with eigenvalue $1$, that is, the corresponding Bohr frequency is $0$.

Property (\ref{eq:sqrtGs}) is equivalent to 
\begin{align}
    G^\alpha(\omega) e^{-\omega/2} = \sqrt{G^\alpha(\omega) \, G^\alpha(-\omega)}
\end{align}


\begin{lemma}[Quantum discriminate]\label{lem:quantum_discriminate}
Let $\cL$ be a detailed balanced Lindbladian with respect to $\sigma$. Then, the quantum discriminate $\cQ^f$ is independent of $f$. It acts as
$\cQ=I+\cK$, where
\begin{align}
    \cK(A) &= \sum_\alpha \sum_\omega
    \sqrt{G^\alpha(\omega) \, G^\alpha(-\omega)} X^\alpha(\omega) A X^{\alpha\dagger}(\omega) 
    - \frac{G^\alpha(\omega)}{2} \big\{
    X^{\alpha\dagger}(\omega) X^{\alpha}(\omega),  A
    \big\}
\end{align}
\end{lemma}

\begin{proof}
To abbreviate, we set $\Omega=\Omega_\sigma^f$. We consider arbitrary but fixed $\alpha$ and $\omega$.  We drop the supscript $\alpha$ when writing $X^\alpha(\omega)$.  

First, we analyze how the similarity transformation acts on the first term in (\ref{eq:cL_alpha}).
Using properties (i) and (ii) in Lemma~\ref{lem:Omega_sigma_f}, we obtain 
\begin{align}
    & \Omega^{-1/2}\big( X(\omega) \Omega^{1/2}(A) X^\dagger(\omega) \big) \\
    &=
    \sum_{k,\ell} \sum_{m,n}
    \big[(1/f)(e^{-(\varepsilon_k - \varepsilon_\ell)}) e^{\varepsilon_\ell}\big]^{1/2}
    \big[f(e^{-(\varepsilon_m - \varepsilon_n)})    e^{-\varepsilon_n}\big]^{1/2}
    P_k X(\omega) P_m A P_n X^\dagger(\omega) P_\ell \\
    &=
    \sum_{(k,m)\in\cJ_\omega} \sum_{(\ell,n)\in\cJ_\omega}
    \big[(1/f)(e^{-(\varepsilon_k - \varepsilon_\ell)}) f(e^{-(\varepsilon_m - \varepsilon_n)})\big]^{1/2}
    e^{-\omega/2}
    P_k X(\omega) P_m A P_n X^\dagger(\omega) P_\ell \\
    &=
    \sum_{(k,m)\in\cJ_\omega} \sum_{(\ell,n)\in\cJ_\omega}
    e^{-\omega/2}
    P_k X(\omega) P_m A P_n X^\dagger(\omega) P_\ell \\
    &=
    e^{-\omega/2} X(\omega) A X^\dagger(\omega)
\end{align}
We have used that $\varepsilon_k-\varepsilon_m=\omega$ and $\varepsilon_\ell-\varepsilon_n=\omega$ must hold or otherwise the term is zero.  But in this case we have $\varepsilon_k-\varepsilon_\ell=\varepsilon_m-\varepsilon_m$. This shows that the arguments for $(1/f)$ and $f$ are the same so that the value of the product inside the square brackets simplifies to $1$.  We now use that $G(\omega) e^{-\omega/2}=\sqrt{G(\omega) G(-\omega)}$.

Second, we analyze how the similarity transformation acts on the two terms arising from the anticommutator. For the first term, we have
\begin{align}
    &
    \Omega^{-1/2} \big( X(\omega) X^\dagger(\omega) \Omega^{1/2}(A) \big) \\
    &=
    \sum_{k,\ell} \sum_{m,n}
    \big[(1/f)(e^{-(\varepsilon_k - \varepsilon_\ell)}) e^{\varepsilon_\ell}\big]^{1/2}
    \big[f(e^{-(\varepsilon_m - \varepsilon_n)}) e^{-\varepsilon_n}\big]^{1/2}
    P_k X(\omega) X^\dagger(\omega) P_m A P_n P_\ell \\
    &=
    \sum_{k,\ell} \sum_m
    \big[(1/f)(e^{-(\varepsilon_k - \varepsilon_\ell)}) \big]^{1/2}
    \big[f(e^{-(\varepsilon_m - \varepsilon_\ell)}) \big]^{1/2}
    P_k X(\omega) X^\dagger(\omega) P_m A P_\ell \\
    &=
    \sum_{k,\ell} 
    \big[(1/f)(e^{-(\varepsilon_k - \varepsilon_\ell)}) \big]^{1/2}
    \big[f(e^{-(\varepsilon_k - \varepsilon_\ell)}) \big]^{1/2}
    P_k X(\omega) X^\dagger(\omega) P_k A P_\ell \\
    &=
    \sum_{k,\ell} 
    P_k X(\omega) X^\dagger(\omega) P_k A P_\ell \\
    &= X(\omega) X^\dagger(\omega) A
\end{align}
We used first that we must have $n=\ell$ and then $m=k$ or otherwise the term is be zero. The second term is analyzed analogously. 
\qed{}
\end{proof}


\begin{theorem}[Quantization of detailed balanced Lindbladians]\label{thm:generalLindbladians}
Let $\cL$ be a detailed balanced Lindbladian $\cL$ with respect to the state $\sigma$ in canonical form as in Definition~\ref{def:Lindbladian_canonical}.
Assume that for each $\alpha$, $G^\alpha(\omega)\le 1$ and there is a block encoding of the hermitian matrix
\begin{align}
    X^\alpha = \sum_\omega X^\alpha(\omega) 
\end{align}
of the form
\begin{align}
    X^\alpha = (\<0^s| \otimes I) S^\alpha (I \otimes |0^s\>)
\end{align}
where $s$ is an integer and $S^\alpha$ is an efficiently implementable reflection that acts on the joint system consisting of an auxilliary $s$-qubit register and the system register.\footnote{This implies that $\| X^\alpha\|\le 1$.} Furthermore, assume that there is an ideal energy estimation for the Hamiltonian $h=-\ln\sigma$.

Then, there exists a corresponding walk unitary $W(\cL)$ such that
its unique eigenvector with eigenphase $0$ is the purification
$|\sigma^{1/2}\> = (\sigma^{1/2} \otimes I) |\Omega\>$ 
and its phase gap is quadratically larger than the spectral gap of the Lindbladian $\cL$. It can be obtained by standard methods for spectral gap amplification applied to the block encoding of the vectorized quantum discriminate matrix $\cQ$ for $\cL$.
\end{theorem}

We provide a proof of the above theorem in Section~\ref{sec:generalLindbladians} by reducing the problem of quantizing detailed balanced Lindbladians to the special case of quantizing Davies generators having reflections as coupling operators.  

\subsection{Detailed balanced quantum channels}

We now briefly explain that our theoretical results and our quantum algorithm also apply to (discrete-time) quantum channels.

Let $\cT$ be a quantum channel in the Schr\"odinger picture. Using the Kraus representation, we can write it as
\begin{align}
    \cT(\rho) = \sum_i {A_i \rho A_i^\dagger}
\end{align}
where the Kraus operators $A_i$ satisfy the condition
\begin{align}\label{eq:Kraus}
    \sum_i A_i^\dagger A_i = I
\end{align}
The quantum channel $\cT^*$ in the Heisenberg picture is simply
\begin{align}
    \cT^*(O) = \sum_i A_i^\dagger O A_i
\end{align}

\begin{definition}[Detailed balanced quantum channel]
Let $\cT$ be a quantum channel. We call $\cT$ detailed balanced with respect to the state $\sigma$ if the quantum channel $\cT^*$ in the Heisenberg picture is hermitian with respect of the $\sigma$-GNS inner product.
\end{definition}

\begin{lemma}[Detailed balanced implies fixed point]
If $\cT$ is detailed balanced with respect to $\sigma$, then $\sigma$ is a fixed point.
\end{lemma}

\begin{proof}
We have
\begin{align}
    \Tr(\sigma A) = \<I,A\>_1 = \< \cT^*(I), A\>_1 = \< I, \cT^*(A) \>_1 = \Tr(\sigma \cT^*(A)) = \Tr(\cT(\sigma) A) 
\end{align}
for all $A$ so that $\cT(\sigma)=\sigma$ holds.
\end{proof}

\begin{remark}
We can associate a Lindbladian $\cL$ to the quantum channel $\cT$ by setting
\begin{align}
    \cL(\rho) = \sum_i A_i \rho A_i^\dagger - \frac{1}{2} \{ A_i^\dagger A_i, \rho \}
\end{align}
Note that $\cL = \cT - I$ because of condition (\ref{eq:Kraus}). It is now clear that $\cT$ is detailed balanced iff the corresponding Lindbladian $\cL$ is detailed balanced. The canonical form for the Lindbladians also carries over for the quantum channels.
\end{remark}


\section{Szegedy walk unitaries for Davies generators with ideal energy estimation}\label{sec:Davies}


Davies generators form an important class of detailed balanced Lindbladians. In this section, we show how to quantize Davies generators whose coupling operators are restricted to be \emph{reflections}. 

The section is organized as follows.
First, we introduce the Davies map $\cL$ for a general Hamiltonian $H$ and an inverse temperature $\beta$, which has the thermal Gibbs state $\sigma = e^{-\beta H}/\Tr\big( {e^{-\beta H} \big)}$
as fixed point. Second, we given an expression for the quantum discriminate in the vectorized representation.  Third, we prove how to efficiently obtain a block encoding of the vectorized quantum discriminate provided that we can perfectly resolve the energies of $H$. Fourth, we discuss how spectral gap amplification techniques applied to the vectorized quantum discriminate enable us to obtain the desired walk unitary for $W(\cL)$ so that the purification $|\sigma^{1/2}\>=(\sigma^{1/2}\otimes I)|\Omega\>$ is its unique eigenvector with eigenphase $0$ and its eigenphase gap is quadratically larger than the spectral gap of $\cL$.


\subsection{Davies generators}


A particular class of Liouvillians \cite{lindblad76,Breuer}, which describes the thermalization of a quantum mechanical subject to thermal noise, is known as Davies generators \cite{Davies,Davies2}. This class of Liouvillians  describes the dissipative dynamics resulting as the weak (or singular) coupling limit from a joined Hamiltonian evolution of a system coupled to a large heat bath. The weak coupling limit permits to consider only the reduced dissipative dynamics, which gives rise to a Markovian semi-group generated by the aforementioned Davies Liovillian. This generator retains the information of the bath temperature $\beta$ and converges to the Gibbs distribution of the system, constructed from the system Hamiltonian $H$. Here, we assume that such a generator is already given in the canonical form. We assume that the system is coupled to a bath through the hermitian operators $S_\alpha$ that drive the dissipative dynamics. If a sufficiently large set of such operators is given, then the Gibbs state is the unique fixed point of the map. We discuss the properties of this generator directly, and are not concerned with the actual derivation. 

We write the system Hamiltonian $H$ as
\begin{align}
	H &= \sum_{k=1}^m \varepsilon_k \Pi_k
\end{align}
where $\varepsilon_k\in [0,1]$ denote the $m$ different eigenvalues and $\Pi_k$ are the projectors onto the corresponding eigenspaces. We choose the Schr{\"o}dinger picture as a convention so that the Davies generator $\cL$ describes the evolution of states.  

The canonical form of the Davies generator $\cL$ is given by   
\begin{align}\label{eq:Davies_canonical}
    \cL(\rho) 
    =
    \sum_\alpha w_\alpha \cL^{\alpha}(\rho)
    =
    \sum_\alpha w_\alpha \sum_\omega
    G^\alpha(\omega) \left(
        {S^\alpha}(\omega) \rho {S^\alpha}^\dag(\omega) - 
        \frac{1}{2} \left\{ {S^\alpha}^\dag(\omega)S^\alpha(\omega),\rho \right\} 
    \right)
\end{align}
The index $\alpha$ enumerates the coupling operators $S^\alpha$ to the environment. 
Note that the weights $w_\alpha$ could be absorbed into the coupling operators $S^\alpha$, but it is more convenient for our purposes to write the Davies generator $\cL$ as a convex sum of Lindbladians $\cL^\alpha$.  For a fixed $\alpha$, the variable $\omega$ runs over all different Bohr frequencies of the system Hamiltonian $H$, that is, the energy differences of the form $\omega = \varepsilon_k - \varepsilon_\ell$. The functions $G^\alpha(\omega)$ correspond to the Fourier transform of the two point correlation functions of the environment, and are bounded \cite{Breuer,Spohn}. These functions depend on the bath operators as well as the thermal state of the bath and encode the equilibrium temperature. We may assume them to be simply given by the Metropolis filter rule or a similar filter as defined by Glauber dynamics. The Lindblad operators $S^\alpha(\omega)$ are given by the Fourier components of the coupling operators $S^\alpha$ which evolve according to the system Hamiltonian 
\begin{align}
	e^{iH t} S^\alpha e^{-i H t } &= \sum_\omega S^\alpha(\omega) e^{i \omega t}.
\end{align}
The Lindblad operators $S^\alpha(\omega)$  induce transitions between the eigenvectors of $H$ with energy $E$ to eigenvectors of $H$ with energy $E+\omega$, and hence act as quantum jump operators, which transfer energy $\omega$ from the system to the bath. 
A direct evaluation shows that the jump operators $S^\alpha(\omega)$ are of  the form
\begin{align}\label{Davi-lind}
    S^\alpha(\omega) &= \sum_{(k, \ell)\in\cJ_\omega} \Pi_k S^\alpha \Pi_\ell
\end{align}
The jump operators $S^\alpha(\omega)$ and filter functions $G^\alpha(\omega)$ satisfy the properties in Definition~\ref{def:Lindbladian_canonical}. Therefore, the Davies generator $\cL$ is detailed balanced with respect to $\sigma$.

We already know that the Gibbs state $\sigma$ is a fixed point of the Davies generator $\cL$.  This follows from Lemma~\ref{lem:DB_implies_fixed_point} since $\cL$ is detailed balanced with respect to $\sigma$.
The next lemma provides a sufficient condition for the Gibbs state $\sigma$ to be the \emph{unique} fixed point. This result was proved in \cite{Spohn}.

\begin{lemma}[Uniqueness of fixed point]
If the coupling algebra\footnote{It is important that a small number of generators will suffice to generate the entire matrix algebra and, thus, guarantee that the fixed point $\sigma$ is unique.} $\calS = \big\langle S^\alpha\big\rangle$ generated by the coupling operators $S^\alpha$ is such that the commutant of $\calS \cup \{H\}$ is only the identity, then the Gibbs state $\sigma$ is the unique fixed point of the Davies generator $\cL$.
\end{lemma}


We are now ready to state our main theorem.

\begin{theorem}[Quantization of Davies generators -- ideal energy estimation]\label{thm:quant_Davies_ideal}
Let $\cL$ be a Davies generator in canonical form as in (\ref{eq:Davies_canonical}) for the Hamiltonian $H$ and inverse temperature $\beta$. Assume that all coupling operators $S^\alpha$ are reflections and all filter functions $G^\alpha$ satisfy $G^\alpha(\omega)\le 1$. 

Then, there exists a corresponding walk unitary $W(\cL)$ such that
its unique eigenvector with eigenphase $0$ is the purification
$|\sigma^{1/2}\> = (\sigma^{1/2} \otimes I) |\Omega\>$ of the Gibbs state $\sigma$
and its phase gap is quadratically larger than the spectral gap of the Davies generator $\cL$. It can be obtained by standard methods for spectral gap amplification applied to the block encoding of the vectorized quantum discriminate matrix.  The block encoding of can implemented exactly provided that we can determine the energies of the Hamiltonian $H$. 
\end{theorem}

We prove this theorem in the next three subsections.


\subsection{Quantum discriminate in the vectorized representation}\label{subsec:vectorized}


For the sake of simplicity, we consider the case that there is only one coupling operator $S$. The general case of multiple coupling operators can be handled easily.

In Section~\ref{sec:Lindbladians}, we showed that a detailed balanced Lindbladian $\cL$ is similar to an operator $\cK$ that has the following form
\begin{align}
    \cK(A) &= \sum_\omega
    \sqrt{G(\omega) \, G(-\omega)} S(\omega) A S^\dagger(\omega) 
    - \frac{G(\omega)}{2} \big\{
    S^{\dagger}(\omega) S(\omega),  A
    \big\}    
\end{align}
Let us recall the key properties of $\cK$:
\begin{itemize}
    \item [(a)] $\cK(\sqrt{\sigma})=0$, 
    \item [(b)] $\cK$ is hermitian with respect to the trace HS inner product, and 
    \item [(c)] $\cK$ does not depend explicitly on the fixed point $\sigma$ (although the similarity transformation was defined using $\sigma$).
\end{itemize}

We will now use the vectorization operation to map the superopertor $\cK$ acting on $\bC^{d\times d}$ to the matrix $\hat{\cK}$ acting $\bC^d\otimes\bC^d$. Recall that the vectorization operation is defined by  
\begin{align}
    \mathrm{vec}(|i\>\<j|) = |i\> \otimes |j\>
\end{align}
and satisfies the useful identity
\begin{align}
    \mathrm{vec}(ABC) = (A\otimes C^T)\mathrm{vec}(B)
\end{align}
for all matrices $A$, $B$, and $C$. The vectorized representation $\hat{\cK}$ of $\cK$ is the unique matrix so that 
\begin{align}
    \mathrm{vec}\big(\cK(A)\big) &= \hat{\cK} \mathrm{vec}(A)
\end{align}
Using the above identity we obtain
\begin{align}\label{eq:hatK}
    \hat{\cK} = \sum_\omega \sqrt{G(\omega)G(-\omega)} S(\omega) \otimes \bar{S}(\omega) - \frac{G(\omega)}{2} \left(S^\dagger(\omega)S(\omega) \otimes I 
    + I \otimes 
    \bar{S}^\dagger(\omega)\bar{S}(\omega)\right)
\end{align}
The properties (a) and (b) manifest themselves in the vectorized representation as follows:
\begin{itemize}
    \item [(a)] The fixed point $\sqrt{\sigma}$ of $\cK$ becomes the fixed point $|\sqrt{\sigma}\>=(\sqrt{\sigma} \otimes I)|\Omega\>$ of $\hat{\cK}$, where $|\Omega\>$ denotes the unnormalized maximally entangled state. Note that $|\sqrt{\sigma}\>$ is a purification of $\sigma$.
    \item [(b)] Hermiticity of $\cK$ with respect to the HS trace inner product on $\bC^{d\times d}$ becomes hermiticity of $\hat{\cK}$ with respect to the standard inner product on $\bC^d\otimes\bC^d$.
\end{itemize}

The \emph{vectorized quantum discriminate} is simply given by $\hat{\cQ}=I+\hat{\cK}$, where $\hat{\cK}$ is as in (\ref{eq:hatK}).  To simplify the notation, we will use $\cQ$ to also denote the quantum discriminate in the vectorization representation.  


\subsection{Block encoding of quantum discriminate}


We construct a block encoding of the vectorized quantum discriminate $\cQ$ assuming ideal energy estimation. We assume as before that there is only one coupling operator $S$ so we may drop the index $\alpha$ when writing the jump operators $S(\omega)$. We define an isometry $T$ and a reflection $R$ to embed the quantum discriminate $\cQ = I + \hat{\cK}$ as follows
\begin{align}
    \cQ &= T^\dagger \; R \; T
\end{align}

Let us explain why we seek to embed $\cQ = I + \hat{\cK}$ instead of $\hat{\cK}$. The ``fixed point'' of  $\hat{\cK}$ is $|\sqrt{\sigma}\>$ with $\hat{\cK} |\sqrt{\sigma}\> = 0$. However, for the Szegedy walk unitary we want to quantize a discriminate matrix satisfying ${\cal Q}|\sqrt{\sigma}\> = |\sqrt{\sigma}\>$. This is accomplished by working with ${\cal Q} = I + \hat{\cK}$. We will also see that $\cQ$ can be embedded in a very natural and elegant way.  

To construct the isometry $T$ and the reflection $R$, we need to define the following building blocks.
\begin{itemize}
    \item {\bf Ideal energy estimation unitaries:}
    \begin{align}
        \Phi       
        &=
        \sum_p \sum_k \Pi_k \otimes |p - \varepsilon_k\>\<p| \\
        \bar{\Phi} 
        &= 
        \sum_p \sum_k \bar{\Pi}_k \otimes |p - \varepsilon_k\>\<p|
    \end{align}
     
     Here $\Phi$ denotes phase estimation of the unitary $e^{-i H}$ where $H=\sum_k \varepsilon_k \Pi_k$, while $\bar{\Phi}$ denotes phase estimation of the unitary $e^{-i \bar{H}}$ where $\bar{H}=\sum_k \varepsilon_k \bar{\Pi}_k$. The variable $p$ runs over all possible values of the energy pointer register. 
     
     Note that both $\Phi$ and $\bar{\Phi}$ constitute idealized energy estimation procedures making it possible to resolve the energies perfectly. 
     
     \item {\bf The coupling and kick operator:}
     
     \noindent
     We consider a coupling and kick operator that serves both as observable and as unitary that can be implemented. We require that it is a reflection so that $S = S^\dagger$ as well as $S^\dagger S = S S^\dagger = I$. We also assume that we can implement $\bar{S}$ (by complex conjugating the elementary gates used in the decomposition of $S$).
     
     \item {\bf The controlled filter rotation:}
     \begin{align}
        W &= \sum_{\omega}  |\omega \>\< \omega| \otimes \hat{G}(\omega)
     \end{align}
     
     where
     \begin{align}
     \hat{G}(w) &= \left(
        \begin{array}{cc}
            \sqrt{1-G(\omega)} & -\sqrt{  G(\omega)} \\
            \sqrt{  G(\omega)} &  \sqrt{1-G(\omega)}
        \end{array}
    \right).
     \end{align}
     
     \item{\bf The sign flip reflection:}
     
     \noindent
     We define a reflection on a single register that maps computational basis 
     states from $|\omega\>$ to $|-\omega\>$ by setting
     \begin{align}
        F &= \sum_{\omega} |-\omega\>\<\omega|
     \end{align}
     This map is a reflection since it squares to the 
     identity.
\end{itemize}

Let us now implement the isometry $T$. It is composed of the two sub-isometries $T_0$ and $T_1$ that are defined in the circuits below. 

\begin{figure}[ht]
\begin{center}
\begin{quantikz}
\lstick{\texttt{filter}} & \qw \gategroup[4, steps=4,style={dashed, rounded corners}, background]{$T_0$} & \qw     & \qw              &  \gate{W}     & \qw & \qw & & \qw \gategroup[4, steps=4,style={dashed, rounded corners}, background]{$T_1$}             & \qw            & \qw                    &  \gate{W}     & \qw \\
\lstick{\texttt{freq}}   & \gate{\Phi}   & \qw       & \gate{\Phi^\dag} &  \ctrl{-1}    & \qw & \qw & & \gate{\bar{\Phi}} & \qw            & \gate{\bar{\Phi}^\dag} &  \ctrl{-1}    & \qw \\
\lstick{\texttt{sys1}}   & \qw           & \qw       & \qw              &  \qw          & \qw & \qw & & \ctrl{-1}         & \gate{\bar{S}} & \ctrl{-1}              &  \qw          & \qw \\
\lstick{\texttt{sys2}}   & \ctrl{-2}     & \gate{S}  & \ctrl{-2}        &  \qw          & \qw & \qw & & \qw               & \qw            & \qw                    &  \qw          & \qw 
\end{quantikz}
\end{center}
\caption{Circuits for the subisometries $T_0$ and $T_1$.}
\label{fig:circuits_T0_T1}
\end{figure}
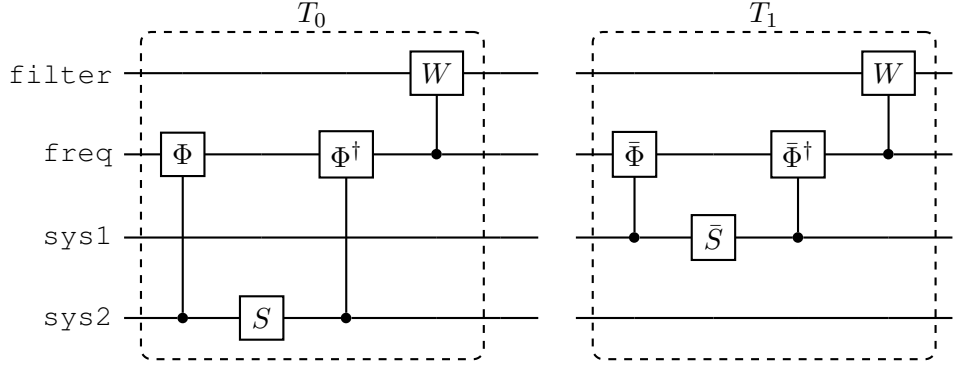

\noindent
There are four registers. The registers \texttt{sys1} and \texttt{sys2} correspond to a ``doubled'' physical register that can encode purifications of quantum states. The register \texttt{freq} encodes the Bohr frequencies of the Hamiltonian. It is prepared in $|0^r\>$. The register \texttt{filter} is a single qubit register on which we preform the controlled rotation $\hat{G}(\omega)$. It is prepared in $|0\>$. The sub-isometry $T_0$ has the kick operator $S$ act on the register \texttt{sys2}, whereas the sub-isometry $T_1$ has the complex-conjugate kick operator $\bar{S}$ act on the register \texttt{sys1}. 

The full isometry $T$ is obtained by coherently adding $T_0$ and $T_1$. To this end, we introducing a single qubit register \texttt{add} that is prepared in $|+\>$ and apply controlled versions of these two respective isometries as show in the circuit below.

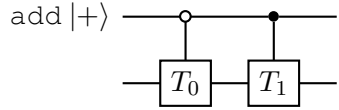
\begin{figure}[ht]
\begin{center}
\begin{quantikz}
\lstick{\texttt{add} $|+\>$} & \octrl{1}  & \ctrl{1}   & \qw \\
                      & \gate{T_0} & \gate{T_1} & \qw 
\end{quantikz}
\end{center}
\caption{Circuit for the full isometry $T$}
\end{figure}

\noindent
The full isometry $T$ is given by
\begin{align}
    T = \frac{1}{\sqrt{2}}\left(T_0 \otimes |0\> + T_1 \otimes |1\> \right)
\end{align}

Before we analyze $T$ in detail, let us first understand the action of the ideal energy estimation operations and kick operators. It suffices to focus on $T_0$ as $T_1$ behaves analogously. We only consider the action on the two ``active'' registers \texttt{sys2} and \texttt{freq} for notional simplicity.

The ideal energy estimation acts as
\begin{align}
\Phi^\dagger \big(S \otimes I \big) \Phi 
    &=  
    \sum_p \sum_k \; \sum_q \sum_\ell \; \Pi_k \, S \, \Pi_\ell \; \otimes \;
    |q\>\<q - \varepsilon_k | p - \varepsilon_\ell\>\<p| \\
    &=  
    \sum_{p} \; \sum_{k, \ell} \; \Pi_k \, S \, \Pi_\ell \; \otimes \;
    | \varepsilon_k -\varepsilon_\ell + p\>\<p| \\
    &=  
    \sum_{p} \sum_{\omega} S(\omega) \otimes \;| \omega  + p\>\<p|
\end{align}
where we used the definition of the ideal phase estimation unitary $\Phi$ as wells as the definition of the jump operators $S(\omega) = \sum_{(k,\ell)\in\cJ_\omega}  \Pi_k \, S \, \Pi_\ell$.

According to the circuits these isometries are defined to be
\begin{align}
    T_0 
    &= 
    W \, \Phi_1^\dagger \, \big(S \otimes I \big) \, \Phi_1 \, 
    \big( I \otimes I \otimes |0\>^r \otimes |0\> \big) \\ 
    T_1 
    &= 
    W \, \bar{\Phi}_2^\dagger \, \big(I \otimes \bar{S} \big) \, \bar{\Phi}_2 \, 
    \big( I \otimes I \otimes |0\>^r \otimes |0\> \big) 
\end{align}
Using the expression for the ideal energy estimation, we obtain
\begin{align}
    T_0 
    &= 
    \sum_{\omega}  S(\omega) \otimes I \otimes |\omega\> \otimes \left(\sqrt{1-G(\omega)}|0\> + \sqrt{G(\omega)} |1\>\right) \\ 
    T_1 
    &= 
    \sum_{\omega} I \otimes \bar{S}(\omega) \otimes |\omega\> \otimes \left(\sqrt{1-G(\omega)}|0\> + \sqrt{G(\omega)} |1\>\right)
\end{align}
for the isometries $T_0$ and $T_1$. Hence, we obtain
\begin{align}
    T 
    &= 
    \frac{1}{\sqrt{2}}\sum_{\omega}  S(\omega) \otimes I \otimes |\omega\> \otimes 
    \left(\sqrt{1-G(\omega)} |0\> + \sqrt{G(\omega)} |1\>\right) \otimes |0\> \\
    &+ 
    \frac{1}{\sqrt{2}} \sum_{\omega}  I \otimes \bar{S}(\omega) \otimes |\omega\> \otimes 
    \left(\sqrt{1-G(\omega)} |0\> + \sqrt{G(\omega)} |1\>\right)\otimes |1\> 
\end{align}
for the full isometry $T_0$.

Before we define the reflection $R$, we recall that the registers in the order from left to right are: \texttt{sys2}, \texttt{sys1}, \texttt{freq}, \texttt{filter}, and \texttt{add} in the above expression for the full isometry $T$.

The reflection operator $R$ is defined by the circuit

\gategroup[2, steps=2, style={dashed, rounded corners}, background]{$R$}   

\begin{figure}[ht]
\begin{center}
\begin{quantikz}
\lstick{\texttt{add}}    & \gate{X}  & \qw      & \qw \\
\lstick{\texttt{filter}} & \ctrl{-1} & \ctrl{1} & \qw \\
\lstick{\texttt{freq}}   & \qw       & \gate{F} & \qw \\
\lstick{\texttt{sys1}}   & \qw       & \qw      & \qw \\[0.25cm]
\lstick{\texttt{sys2}}   & \qw       & \qw      & \qw \\
\end{quantikz}
\end{center}
\caption{Circuit for reflection $R$.}
\end{figure}
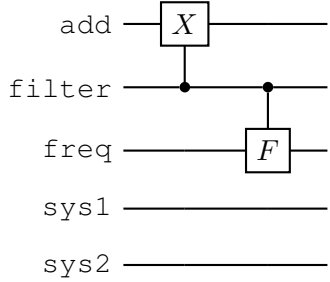

\noindent
$R$ applies the sign flip $\omega\leftrightarrow -\omega$ in \texttt{freq} and the bit flip $0\leftrightarrow 1$ in \texttt{add} iff  the control qubit \texttt{filter} is in the state $|1\>$. The expression for $R$ is
\begin{align}
   R
   &=
   I\otimes I \otimes \; I  \otimes |0\>\<0| \otimes I \\
   &+
   I\otimes I \otimes F \otimes |1\>\<1| \otimes X  
\end{align}

The remainder of this section is devoted to proving that the desired vectorized quantum discriminant $\cQ$ is given by the block encoding
\be
\cQ = T^\dagger R T
\ee

We begin by computing the product $R T$. To this end, it is convenient to rearrange the terms of $T$ as follows
\begin{align}    
    T
    &=
    \frac{1}{\sqrt{2}} \sum_{\omega} \sqrt{1 - G(\omega)} \; S(\omega) \otimes I \otimes |\omega\> \otimes 
    |0\>  \otimes |0\> \\
    &+
    \frac{1}{\sqrt{2}} \sum_{\omega} \sqrt{1 - G(\omega)} \; I \otimes \bar{S}(\omega) \otimes |\omega\> \otimes 
    |0\> \otimes |1\> \\
    &+ 
    \frac{1}{\sqrt{2}} \sum_{\omega} \sqrt{G(\omega)} \; S(\omega) \otimes I \otimes |\omega\> \otimes 
    |1\> \otimes |0\> \\
    &+ 
    \frac{1}{\sqrt{2}} \sum_{\omega}  \sqrt{G(\omega)} \; I \otimes \bar{S}(\omega) \otimes |\omega\> \otimes 
    |1\> \otimes |1\> 
\end{align}

Using the above expression for $T$, it is seen that the product $RT$ is given by
\begin{eqnarray}    
    R T
    &=&
    \frac{1}{\sqrt{2}} \sum_{\omega} \sqrt{1-G(\omega)} \; S(\omega) \otimes I \otimes |\omega\> \otimes 
    |0\>  \otimes |0\> \\
    &+&
    \frac{1}{\sqrt{2}} \sum_{\omega} \sqrt{1-G(\omega)} \; I \otimes \bar{S}(\omega) \otimes |\omega\> \otimes 
    |0\> \otimes |1\> \\
    &+& 
    \frac{1}{\sqrt{2}} \sum_{\omega} \sqrt{G(\omega)} \; S(\omega) \otimes I \otimes |-\omega\> \otimes 
    |1\> \otimes |1\> \\
    &+& 
    \frac{1}{\sqrt{2}} \sum_{\omega}  \sqrt{G(\omega)} \; I \otimes \bar{S}(\omega) \otimes |-\omega\> \otimes 
    |1\> \otimes |0\> 
\end{eqnarray}

Observe the sign flip $F|\omega\> = |-\omega\>$ and the bit flip $|0\>\leftrightarrow |1\>$ in the final two lines. 
Using the above expression for $RT$, it is now seen that the quantum discriminate ${\cal Q}=T^\dagger (R T)$ is given by 
\begin{align}
    {\cal Q} 
    &= 
    \frac{1}{2}\sum_{\omega} \; \big( 1 - G(\omega) \big) \; S^\dagger(\omega) S(\omega) \otimes I  \\
    &+ 
    \frac{1}{2}\sum_{\omega} \; \big( 1 - G(\omega) \big) \; I \otimes \bar{S}^\dagger(\omega)\bar{S}(\omega)  \\
    &+ 
    \frac{1}{2}\sum_{\omega, \nu} \; \sqrt{G(\omega) \, G(\nu)} \; S(\omega) \otimes \bar{S}^\dagger(\nu) \;
    \delta_{-\omega, \nu} \\
    &+ 
    \frac{1}{2}\sum_{\omega, \nu} \; \sqrt{G(\omega) \, G(\nu)} \; S^\dagger(\nu) \otimes \bar{S}(\omega) \;
    \delta_{-\omega, \nu} 
\end{align}

To further simplify the above expression for ${\cal Q}$, we use the following two elementary properties of the jump operators $S(\omega)$: (i) $S^\dagger(-\omega) = S(\omega)$ and (ii) $\sum_{\omega} S^\dagger(\omega)S(\omega) = I$.  Note that is essential for (ii) that $S$ is a reflection. For the sake of completeness, we provide a proof for (ii).
\begin{align}
    \sum_\omega S^{\dagger}(\omega) S(\omega) 
    &=
    \sum_\omega 
    \left( \sum_{(k,\ell) \in \cJ_\omega} P_\ell S^{\dagger} P_k \right) 
    \left( \sum_{(m, n) \in \cJ_\omega} P_m S P_n \right) \\
    &=
    \sum_\omega 
    \sum_{(k,\ell) \in \cJ_\omega} \sum_{n : (k, n) \in \cJ_\omega} P_\ell S^{\dagger} P_k S P_n \\
    &=
    \sum_\omega 
    \sum_{(k,\ell) \in \cJ_\omega} P_\ell S^{\dagger} P_k S P_\ell \\
    &=
    \sum_{k,\ell} P_\ell S^{\dagger} P_k S P_\ell = S^{\dagger} S = I
\end{align}

We now obtain the final desired expression for the vectorized quantum discriminate ${\cal Q}$  
\begin{align}
    {\cal Q} 
    &= 
    I \otimes I + 
    \sum_{\omega}\; \sqrt{G(\omega)G(-\omega)} \; S(\omega) \otimes \bar{S}(\omega) 
    - \frac{G(\omega)}{2} \left( 
        S^\dagger(\omega)S(\omega)\otimes I 
        + I \otimes \bar{S}^\dagger(\omega) \bar{S}(\omega) 
    \right)
\end{align}

Using the expression of $\hat{\cal K}$ in (\ref{eq:hatK}), we see that the definition of the vectorized quantum discriminate $\cQ=I+\hat{\cK}$ coincides with the constructed block encoding ${\cal Q} = T^\dagger R T$. 


\subsection{Gap amplification of quantum discriminate}

In the previous two subsections, we considered the case of a single coupling operator $S$. Let us now show how to handle the general case of multiple coupling operators $S^\alpha$. Let $w_\alpha$ denote the corresponding weights.

The results in the previous two subsections show how to construct a reflection $T^\alpha$ for each $\alpha$ such that $\cQ^\alpha=T^{\alpha\dagger} R T^\alpha$.
For each $\alpha$, we set $\theta_\alpha=\arccos(w_\alpha)/2$ so that $\cos^2\theta_\alpha-\sin^2\theta_\alpha=\cos(2\theta)=w_\alpha$. We define the isometry
\begin{align}
    T = \sum_\alpha T^\alpha \otimes |\alpha\> \otimes 
    \big( \cos\theta_\alpha |0\> + \sin\theta_\alpha |1\> \big)
\end{align}
Let us use the abbreviation $R$ to denote the reflection $R\otimes I \otimes Z$, where $I$ acts on the register holding $|\alpha\>$ and $Z$ acts on the qubit containing the states $\cos\theta_\alpha |0\> + \sin\theta_\alpha |1\>$.

We now see that
\begin{align}
    \cQ = T^\dagger R T
\end{align}
The desired walk unitary is given by 
\begin{align}
    \cW(\cL) = R (2\Pi - I)
\end{align}
where $\Pi$ denotes the projector onto the image of the isometry $T$.  This is a very well-known result in quantum algorithms.  For the sake of completeness, we have include a proof of this result in the appendix.

We point out that the quantum singular value transformation (QSV) \cite{gilyen19quant_sing_val} enables us to directly construct the projector $|\sqrt{\sigma}\>\<\sqrt{\sigma}|$ onto the eigenvector of $\cQ$ with eigenvalue $1$, that is, the purification $|\sqrt{\sigma}\>$, by invoking the block encoding $O(1/\sqrt{\Delta})$ many times, where $\Delta$ is the spectral gap of $\cQ$ (which is equal to the spectral gap of the Lindbladian). Being able to implement this projector (or the corresponding reflection) is what is really needed for quantum algorithms.  Before the advent of QSV, the usual approach to realize this projector was to perform quantum phase estimation of the walk unitary to distinguish the eigenphase $0$ from the other eigenphases.  


\section{Szegedy walk unitaries for Davies generators with approximate energy estimation}\label{sec:Davies_approximate}

In the preceeding section, we assumed ideal energy estimation to explain our method for quantizing Lindbladians. We now remove this unphysical assumption and work with the method \cite{rall2021} for energy estimation. To state the performance of this energy estimation method, we need to first introduce the following rounding promise.

\begin{definition}[Rounding promise]
Let $r$ be a positive integer and $\alpha\in(0,1)$. A hermitian matrix $H$ with spectral decomposition
\begin{align}
    H = \sum_k \varepsilon_k \Pi_k
\end{align}
is said to satisfy an $(r, \alpha)$-\emph{rounding promise} if for all its eigenvalues $0\le \varepsilon_k < 1$ and for all $x\in\{0,\ldots,2^r\}$:
\begin{align}
    \left| \varepsilon_k - \frac{x}{2^r}\right| \ge \frac{\alpha}{2^{r+1}}
\end{align}
\end{definition}

The rounding promise effectively disallows the eigenvalues $\varepsilon_k$ of $H$ from certain sub-intervals of $[0, 1)$. The definition above is chosen such that the total length of these disallowed subintervals is $\alpha$, regardless of the value of $r$. We acknowledge that guaranteeing a rounding-promise would demand a tremendous amount of knowledge about energies of Hamiltonians. It is not expected that many Hamiltonians in practice provably satisfy such a strong promise. But it seems likely that the energy estimation method still works sufficiently well even if the rounding promise is violated so it can be used a subroutine for quantizing Davies generators. As shown below, the rounding promise enables us to precisely analyze the method for quantizing Davies generators. 


\begin{theorem}[Approximate energy estimation]\label{thm:energy_estimation}
Assume that we are given access to a block encoding of a hermitian matrix $H=\sum_k \varepsilon_k \Pi_k$ satisfying an $(r, \alpha)$-rounding promise. Then, we can realize a quantum map $\Phi$ that is $\delta$-close in $\diamond$-norm to the unitary energy estimation map $U$ that acts on states\footnote{Note that action of $U$ on states $|\psi\> \otimes |z\>$ with $z\neq 0^n$ can be arbitrary.} $|\psi\> \otimes |0^r\>$ as
\begin{align}
    U \big( |\psi\> \otimes |0^n\> \big) =
    \sum_{k} \Pi_k |\psi\> \otimes |\lfloor \varepsilon_k 2^r \rfloor\>
\end{align}
The query complexity with respect to the block encoding of $H$ is bounded by
\begin{align}\label{eq:query_complexity}
    O\Big( \alpha^{-1}\log(\delta^{-1})
    \big(2^r + \log(\alpha^{-1}\big)\Big)
\end{align}
\end{theorem}
The above result is proved in \cite[Corollary 16. Improved Energy Estimation]{rall2021}. We emphasize that without any rounding promise \emph{any} method for energy estimation will necessarily produce states close to 
\begin{align}
    \sum_k \Pi_k \otimes \sum_x c_{k,x} |x\>
\end{align}
where the magnitude of $c_k,x$ depends on the distance between $\varepsilon_k 2^r$ and $x$.  In words, instead of obtaining a single estimate for each $k$, we would obtain a superposition of estimates. This is unavoidable as discussed in \cite[1 Preliminaries]{rall2021}.  Roughly speaking, the reason is a follows: there are only a finite number of possible pointer values for the energies, whereas the energies can be arbitrary values in the interval $[0, 1)$. These superpositions greatly complicate the analysis of quantum algorithms (and are often disregarded in the literature).

We also note that is not possible to improve the scaling with $2^r$. It is known that highly precise energy estimation would make it possible to solve PSPACE-hard problems \cite{pspace}.


Using the above method for energy estimation, we can approximately determine the Bohr frequencies and also approximately realize the jump operators.

\begin{lemma}[Bohr frequency estimation and approximate jump operators]
Given a reflection $S$ and access to a block encoding of a hermitian matrix $H=\sum_k \varepsilon_k \Pi_k$ satisfying an $(r, \alpha)$-\emph{rounding promise} we can implement a quantum map $\Psi$ that is $\delta$-close in $\diamond$-norm to the unitary $V$ that acts on states $|\psi\>\otimes |0^r\> \otimes |0^r\>$ as
\begin{align}\label{eq:bohr_estimation}
    V (|\psi\> \otimes |0^r\> \otimes |0\>^r) &=
    \sum_{k,\ell} \Pi_k S \Pi_\ell |\psi\> \otimes 
    |\lfloor \varepsilon_k 2^r \rfloor - 
    \lfloor \varepsilon_\ell 2^r \rfloor \> \otimes |0^r\>
\end{align}
The query complexity with respect to the block encoding of $H$ is three times the query complexity of the method for energy estimation in (\ref{eq:query_complexity}).
\end{lemma}
\begin{proof}
The unitary $V$ can be implemented using the unitary $U$ in Theorem~\ref{thm:energy_estimation} twice and $U^\dagger$ once and a unitary that performs integer subtraction. The circuit is composed of the following sub-circuits:
\begin{itemize}
    \item apply $U$ to compute $\lfloor \varepsilon_k 2^r \rfloor$ into the energy register $A$
    \item apply the reflection $S$
    \item apply $U$ to compute $\lfloor \varepsilon_\ell 2^r \rfloor$ into the energy register $B$
    \item apply a unitary that subtracts the value stored in $B$ from the value stored $A$
    \item apply $U^\dagger$ to uncompute $\lfloor \varepsilon_\ell 2^r\rfloor$ in $B$
\end{itemize}
\end{proof}


\begin{theorem}[Quantizing Davies generators with approximate energy estimation]
Given access to a block encoding of a hermitian matrix $H=\sum_k \varepsilon_k \Pi_k$ satisfying an $(r, \alpha)$-\emph{rounding promise}, we can realize a quantum map that is $\delta$-close in $\diamond$-norm to a block encoding of the quantum discriminate matrix $\tilde{\cQ}$ corresponding to the Davies generator $\tilde{\cL}$ for the approximate Hamiltonian 
\begin{align}
    \tilde{H} &= \sum_k \frac{\lfloor \varepsilon_k 2^r\rfloor}{2^r} \, \Pi_k
\end{align}
The query complexity with respect to the block encoding of $H$ is bounded by
\begin{align}\label{eq:query_complexityin_thm}
    O\Big( \alpha^{-1}\log(\delta^{-1})
    \big(2^r + \log(\alpha^{-1}\big)\Big)
\end{align}
\end{theorem}


\begin{proof}
The following observation is simple yet it greatly improves the analysis. Observe that the unitary $V$ in (\ref{eq:bohr_estimation}) can be interpreted as performing \emph{ideal} Bohr frequency estimation for the \emph{approximate} Hamiltonian $\tilde{H}$.  We can write its action\footnote{We omit the energy register $B$ whose value is alway reversibly reset to the initial state $|0^r\>$.} as
\begin{align}
    V (|\psi\> \otimes |0^r\>) &=
    \sum_{\tilde{\omega}} S(\tilde{\omega}) |\psi\> \otimes 
    |\omega\>
\end{align}
where $\tilde{\omega}$ runs over the Bohr frequencies of the approximate Hamiltonian $\tilde{H}$ and $S(\tilde{\omega})$ denote the corresponding jump operators for $\tilde{H}$. Now everything can be analyzed exactly as in the previous section by replacing $H$ with $\tilde{H}$, $\omega$ with $\tilde{\omega}$, and $S(\tilde{\omega})$ with $S(\omega)$, etc. Of course, we also have to take into account that $V$ (or its complex conjugate version) cannot be realized exactly, but is only approximated by a quantum map that is $\delta$-close in $\diamond$-norm to the unitary map $V$.
The query complexity in (\ref{eq:query_complexityin_thm}) follows from the discussion in the lemma. \qed{}
\end{proof}

\begin{remark}
To summarize, we accomplish the following. Given a block encoding of a Hamiltonian $H$ with an $(r, \alpha)$ rounding promise, we can quantize a Davies generator $\tilde{\cL}$ for the approximate Hamiltonian $\tilde{H}$. Note that the eigenvector of the corresponding approximate walk unitary $\tilde{W}$ with eigenphase $0$ is the purification
\begin{align}
    |\tilde{\sigma}^{1/2}\> 
    &= 
    (\tilde{\sigma}^{1/2} \otimes I) |\Omega\> 
\end{align}
where $\tilde{\sigma}$ denotes the thermal state for $\tilde{H}$ at the inverse temperature $\beta$. It can be shown that 
\begin{align}
    \<\sigma^{1/2}|\tilde{\sigma}^{1/2}\> 
    &=
    \frac{1}{\sqrt{Z} \sqrt{\tilde{Z}}}
    \sum_k \sqrt{ e^{-\beta(\varepsilon_k + \tilde{\varepsilon}_k)}} \, \Tr(\Pi_k) \\
    &\ge 
    \frac{Z}{\tilde{Z}} \\
    &\ge
    e^{-\beta/2^r} \ge 1 - \beta/2^r
\end{align}
where we have used that $\tilde{\varepsilon}_k = \lfloor \varepsilon_k 2^r\rfloor/2^r < \varepsilon_k$ for the first inequality. 
To ensure that the overlap between the ideal and approximate purifications is at least $1-\kappa$, we need to have $r\ge \log_2(\beta/\kappa)$.

It is also important to point out that in general the correct Davies generator $\cL$ and the approximate Davies generator $\tilde{\cL}$ will have different spectral gaps. The difference between these gaps will decrease with increasing $r$.
\end{remark}


\section{Szegedy walk unitaries for detailed balanced Lindbladians}\label{sec:generalLindbladians}

We extend the quantization method for Davies generators with reflections as coupling operators to detailed balanced Lindbladians. We provide a proof of Theorem~\ref{thm:generalLindbladians}.

Let $\cL$ be a detailed balanced Lindbladian with fixed point $\sigma$ in canonical form as in Definition~\ref{def:Lindbladian_canonical}. It will be convenient to use $\theta$ instead of $\omega$ to denote the Bohr frequencies of the Hamiltonian $h$.
For each Bohr frequency $\theta$, we now define the operator $\Lambda_\theta$ by
\begin{align}
    \Lambda_\theta(A) = \sum_{(k,\ell)\in\cJ_\theta} \Pi_k A \Pi_\ell
\end{align}
where $A$ is an arbitrary matrix. We now argue that for each $\alpha$ there always exists a hermitian matrix $X^\alpha$ such that the jump operators $X^\alpha(\theta)$ for all $\theta$ arise from $X^\alpha$ by the projection
\begin{align}
    \Lambda_\theta(X^\alpha) = X^{\alpha}(\theta)
\end{align}
This is seen as follows. Observe that the condition in $(\ref{eq:eig_modular})$ in the canonical form for detailed balanced Lindbladians means that the jump operators $X^\alpha(\theta)$ are all eigenvectors of the so-called \emph{modular operator} $\Delta_\sigma$ defined by $\Delta_\sigma(A)=\sigma A \sigma^{-1}$. The modular operator is a hermitian operator with respect to the Hilbert-Schmidt trace inner product. This implies that these eigenvectors $X^\alpha(\theta)$ are necessarily all orthogonal to each other as the eigenvalues $e^{-\theta}$ are all different. It turns out that the operators $\Lambda_\theta$ are precisely the orthogonal projectors onto the eigenspaces of $\Delta_\sigma$ with eigenvalues $e^{-\theta}$. The desired matrix $X^\alpha$ is given by $X^\alpha = \sum_\omega X^\alpha(\theta)$.

We need to assume that the hermitian matrices $X^\alpha$ satisfy norm constraint $\|X^\alpha\|\le 1$.  This norm constraint implies that there exist reflections $S^\alpha$ that are block encodings of the matrices $X^\alpha$ of the form 
\begin{align}
    X^\alpha &= (\<0^s| \otimes I) S^\alpha (|0^s\> \otimes I) 
\end{align}

Using these block encodings, we generalize the circuits for $T_0$ and $T_1$ as follows. To simplify the discussion, we consider the case where $\alpha$ is arbitrary but fixed. We drop the index $\alpha$ in the following. 

\begin{figure}[ht]
\begin{center}
\begin{adjustbox}{width=0.9\textwidth}
\begin{quantikz}
\lstick{\texttt{filter}} & \qw \gategroup[7, steps=6,style={dashed, rounded corners}, background]{$T_0$} & \qw & \qw & \qw & \qw & \gate{W} & \qw & & \qw \gategroup[7, steps=6,style={dashed, rounded corners}, background]{$T_1$} & \qw & \qw & \qw & \qw & \gate{W} & \qw \\
\lstick{\texttt{block}}  & \gate{C}  & \qw           & \qw               & \gate{C}  & \qw                 & \octrl{-1} & \qw & & \gate{C}  & \qw               & \qw                     & \gate{C}  & \qw                       & \octrl{-1} & \qw \\
\lstick{\texttt{freq}}   & \qw       & \gate{\Phi}   & \qw               & \qw       & \gate{\Phi^\dagger} & \ctrl{-1}  & \qw & & \qw       & \gate{\bar{\Phi}} & \qw                     & \qw       & \gate{\bar{\Phi}^\dagger} & \ctrl{-1}  & \qw \\
\lstick{\texttt{sys1}}   & \qw       & \qw           & \qw               & \qw       & \qw                 & \qw        & \qw & & \qw       & \ctrl{-1}         & \gate[wires=2]{\bar{S}} & \qw       & \ctrl{-1}                 & \qw        & \qw \\ 
\lstick{\texttt{anc1}}   & \qw       & \qw           & \qw               & \qw       & \qw                 & \qw        & \qw & & \ctrl{-3} & \qw               & \qw                     & \ctrl{-3} & \qw                       & \qw        & \qw \\
\lstick{\texttt{sys2}}   & \qw       & \ctrl{-3}     & \gate[wires=2]{S} & \qw       & \ctrl{-3}           & \qw        & \qw & & \qw       & \qw               & \qw                     & \qw       & \qw                       & \qw        & \qw \\ 
\lstick{\texttt{anc2}}   & \ctrl{-5} & \qw           & \qw               & \ctrl{-5} & \qw                 & \qw        & \qw & & \qw       & \qw               & \qw                     & \qw       & \qw                       & \qw        & \qw 
\end{quantikz}
\end{adjustbox}
\end{center}
\caption{Circuits for subisometries $T_0$ and $T_1$.}
\end{figure}
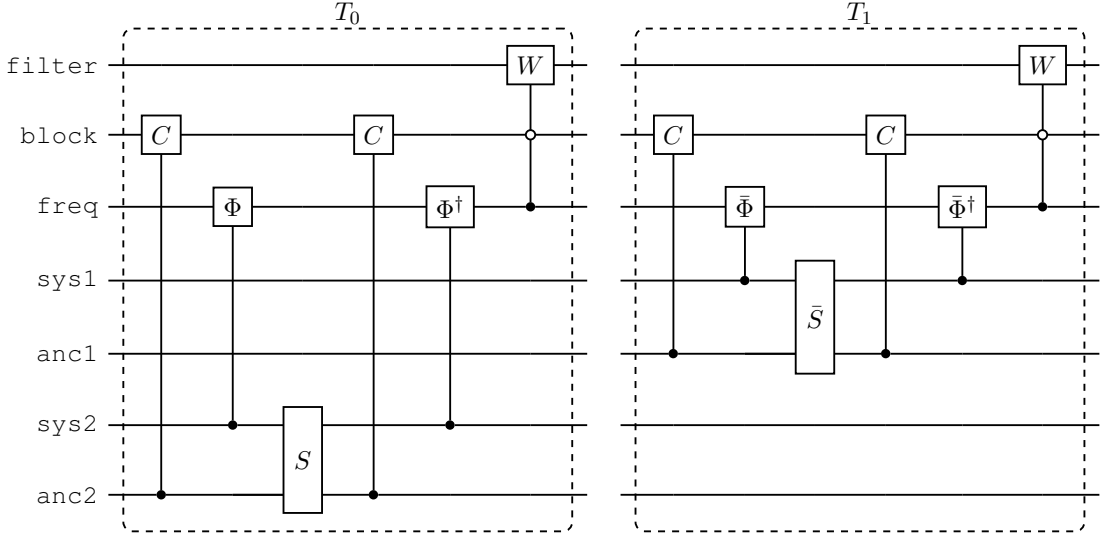

\noindent
It suffices to explain the circuit for $T_0$ as $T_1$ works analogously. We have added the new register \texttt{block}, \texttt{anc1}, and \texttt{anc2} (compare to the circuits in Figure~\ref{fig:circuits_T0_T1}.)
The two system registers \texttt{sys1} and \texttt{sys2} have the ancilla registers \texttt{anc1} and \texttt{anc2} attached, respectively. Both ancilla registers are prepared in $|0^s\>$. The state (either $0$, $1$, $2$) of the \texttt{block} register indicates which block of $S$ is selected. Let $P_0=|0^s\>\<0^s|$ and $P_1=I - P_0$. The state $0$ indicates the diagonal block $P_0 S P_0$ encoding $X$, the state $1$ the off-diagonal blocks $P_0 S P_1$ or $P_1 S P_0$, and the state $2$ the diagonal block $P_1 S P_1$. Let $C$ denote the cyclic shift operator $|x\> \mapsto |x + 1 \pmod{3}\>$. The controlled cyclic shift operator is defined by $P_0 \otimes I + P_1 \otimes C$. The controlled filter rotation $W$ is applied only if the state of \texttt{block} is $|0\>$, which means that the state of \texttt{anc2} was $|0^s\>$ before and after the application of $S$. The gate $\Phi$ and $\Phi^\dagger$ denotes phase estimation with respect to the Hamiltonian $h$ acting on \texttt{sys2} and are used to record the Bohr frequency in the register \texttt{freq}.

As before we use $\theta$ (instead of $\omega)$ to denote the Bohr frequencies of the Hamiltonian $h$.
Define the ``macro'' jump operator $S(\theta)$ by
\begin{align}
    S(\theta) &= 
    \sum_{(k, \ell)\in\cJ_\theta} 
        \big( I \otimes \Pi_k \big) 
        S
        \big( I \otimes \Pi_\ell \big)
\end{align}
where the projectors $\Pi_k$ and $\Pi_\ell$ act on \texttt{sys2}.
For $c\in\{0,1,2\}$, define the ``micro'' jump operators $S(c, \theta)$ acting on the tensor product of \texttt{anc1} and \texttt{sys1} by
\begin{align}
    S(0, \theta) &= (P_0 \otimes I) S(\theta) (P_0 \otimes I) \\
    S(1, \theta) &= (P_0 \otimes I) S(\theta) (P_1 \otimes I) + (P_1 \otimes I) S(\theta) (P_0 \otimes I) \\
    S(2, \theta) &= (P_1 \otimes I) S(\theta) (P_1 \otimes I) 
\end{align}
Here the projectors $P_0$ and $P_1$ act on \texttt{anc2}. Observe that that for all $\theta$ we have
\begin{align}\label{eq:enc_Xtheta}
    \<0^s| \otimes I) S(0, \theta) (|0^s\>\otimes I) = X(\theta).
\end{align}
The corresponding filter functions $G(c, \theta)$ are defined by
\begin{align}
    G(0, \theta) &= G(\theta) \\
    G(1, \theta) &= 0 \\
    G(2, \theta) &= 0 
\end{align}

Now assume that \texttt{block} and \texttt{freq} are properly initialized. We assume that the ancilla registers \texttt{anc1} and \texttt{anc2} are arbitrary. We obtain as the explicit expressions for the subisometries $T_0$ and $T_1$
\begin{eqnarray}
    T_0 
    &=& 
    \sum_{c, \theta}  S(c, \omega) \otimes I \otimes |c, \theta\> \otimes \left(\sqrt{1-G(c, \theta)}|0\> + \sqrt{G(c, \theta)} |1\>\right) \\ 
    T_1 
    &=& 
    \sum_{c, \theta} I \otimes \bar{S}(c, \theta) \otimes |c, \theta\> \otimes \left(\sqrt{1-G(c, \theta)}|0\> + \sqrt{G(c, \theta)} |1\>\right)
\end{eqnarray}
Note that the first (from the left) tensor component refers to \texttt{anc2} and \texttt{sys2} and the second tensor component to \texttt{anc1} and \texttt{sys1}, and the third tensor component to \texttt{block} and \texttt{bohr}.

The full isometry $T$ is obtained by coherently adding the above two sub-isometries $T_0$ and $T_1$ and is given by
\begin{eqnarray}
    T 
    &=& 
    \sum_{c, \theta}  S(c, \omega) \otimes I \otimes |c, \theta\> \otimes \left(\sqrt{1-G(c, \theta)}|0\> + \sqrt{G(c, \theta)} |1\>\right) \otimes |0\> \\
    &+& 
    \sum_{c, \theta} I \otimes \bar{S}(c, \theta) \otimes |c, \theta\> \otimes \left(\sqrt{1-G(c, \theta)}|0\> + \sqrt{G(c, \theta)} |1\>\right) \otimes |1\>
\end{eqnarray}

We call a pair of the form $(c, \theta)$ an \emph{extended} Bohr frequency, where $c\in\{0,1,2\}$ and $\theta$ is a Bohr frequency of $h$. We reserve the symbol $\omega$ to exclusively denote extended Bohr frequencies. For an arbitary extended Bohr frequency $\omega=(c, \theta)$, we define $-\omega=(c, -\theta)$. Observe that this definition make sense as the property (i) 
\begin{align}
    S(-\omega) = S^\dagger(\omega)
\end{align}
holds. It is straightforward to check that the property (ii)
\begin{align}
    \sum_\omega S^\dagger(\omega) S(\omega) = I
\end{align}
holds as well, where $I$ denotes the identity acting on the tensor product \texttt{anc1} and \texttt{sys1}.

Using the definition of $-\omega$ for extended Bohr frequencies, we now define the reflection operator $R$ as
\begin{align}
   R
   &=
   I\otimes I \otimes \; I  \otimes |0\>\<0| \otimes I \\
   &+
   I\otimes I \otimes F \otimes |1\>\<1| \otimes X  
\end{align}
Observe that $R$ applies the flip $\omega=(c, \theta)\leftrightarrow -\omega=(c, -\theta)$ and the bit flip $0\leftrightarrow 1$ whenever the control qubit \texttt{filter} is in the state $|1\>$.

The entire analysis of the method for quantizing Davies generators now carries over because properties (i) and (ii) hold. 
We obtain a block encoding of the matrix
\begin{align}
    I \otimes I + 
    \sum_{\omega}\; \sqrt{G(\omega)G(-\omega)} \; S(\omega) \otimes \bar{S}(\omega) 
    - \frac{G(\omega)}{2} \left( 
        S^\dagger(\omega)S(\omega)\otimes I 
        + I \otimes \bar{S}^\dagger(\omega) \bar{S}(\omega) 
    \right)
\end{align}
This is not yet the desired quantum discriminate $\cQ$ of the Lindbladian $\cL$ because the ancilla register have not been specialized. The first tensor component corresponds to \texttt{anc1} and \texttt{sys1} and the second tensor component to \texttt{anc2} and \texttt{sys2}. Observe now that we can drop all terms with $c=1,2$ because the corresponding filter functions $G(c, \theta)$ are $0$.  Therefore, if we initialize both ancilla registers \texttt{anc1} and \texttt{anc2} in $|0^s\>$, we obtain the desired block encoding
\begin{align}
    \cQ
    &= 
    I \otimes I + 
    \sum_{\theta}\; \sqrt{G(\theta)G(-\theta)} \; X(\theta) \otimes \bar{X}(\theta) 
    - \frac{G(\theta)}{2} \left( 
        X^\dagger(\theta) X(\theta)\otimes I 
        + I \otimes \bar{X}^\dagger(\theta) \bar{X}(\theta) 
    \right)
\end{align}
where we used 
using the identity in (\ref{eq:enc_Xtheta}).

The above construction shows that any detailed balanced Lindbladian can be written as a Davies generator in a larger Hilbert space and all coupling operators of the Davies generator are reflections. It enables us to reduce the problem of quantizing detailed balanced Lindbladians to the special case of Davies generators. 


\subsection*{Acknowledgements}
P.~W.~and K.~T.~would like to thank Patrick Rall for discussions on energy estimation and Vojt\v{e}ch Havl\'i\v{c}ek for comments on the manuscript.  We would also like to thank Srinivasan Arunachalam, Vojt\v{e}ch Havl\'i\v{c}ek, Giacomo Nannicini, and Mario Szegedy for discussions on classical and quantum walks.  Both authors acknowledge support from the IBM Research Frontiers Institute and the ARO Grant W911NF-20-1-0014.


\newcommand{\etalchar}[1]{$^{#1}$}


\section{Appendix}\label{sec:appendix}

The appendix is an adaptation of \cite[17.2 How to quantize a Markov chain]{childs}.

\begin{definition}[Quantization]
Let $T\in\bC^{N\times M}$ an isometry from $\bC^M$ to $\bC^N$ with $\mathrm{im}(T)=\cA$\footnote[4]{The dimension of the subspace $\cA$ is necessarily equal to $M$.}, $\Pi\in\bC^{N\times N}$ an orthogonal projector with $\mathrm{im}(T) = \cA$, and $S\in\bC^{M\times M}$ a reflection. Let $Q\in\bC^{M\times M}$ be a hermitian matrix, $|\varphi_1\>,\ldots,|\varphi_M\>$ an orthonormal basis of $\bC^M$ consisting of eigenvectors of $Q$ with eigenvalues $\lambda_1=1\gneqq \lambda_2 = 1 - \Delta \ge \lambda_3 \ge \ldots \ge \lambda_M \gneqq -1$.  We refer to $\Delta$ as the classical gap.

We define the quantization of $Q$ to be
\begin{equation}
U=S(2\Pi - I)\in\bC^{N\times N}
\end{equation}
provided that the isometry $T$ and reflection $S$ satisfy the condition
\begin{equation}
    T^\dagger S T = Q.
\end{equation}
We refer to the value $\theta=\arccos(1-\Delta)$ as the quantum gap of the quantization $U$ of $Q$.

For $j\in\{1,\ldots,M\}$, define the states
\begin{equation}
    |\chi_j\>=T|\varphi_j\> \in \bC^N.
\end{equation}
Let $\mathcal{B}$ be the subspace $\mathcal{A}+S\mathcal{A}$, where $S\mathcal{A}=\{ S|\chi\> : |\chi\>\in \mathcal{A}\}$, and $\mathcal{B}^\perp$ the orthogonal complement of $\mathcal{B}$.
\end{definition}

\begin{lemma}
The subspace spanned $|\chi_1\>,\ldots,|\chi_M\>$ coincides with the subspace $\cA$.  
\end{lemma}


\begin{proof}
We have $\sum_{j\in [M]} |\chi_j\>\<\chi_j| = \sum_{j\in [M]} T |\varphi_j\>\<\varphi_j| T^\dagger = T T^\dagger = \Pi$. \qed{}
\end{proof}


\begin{theorem}[Spectrum of quantization]
The subspace $\mathcal{B}$ and its orthogonal complement $\mathcal{B}^\perp$ are invariant under $U$.  The spectrum of $U$ restricted to $\mathcal{B}$ is as follows:
\begin{enumerate}
\item For $j=1$, the one-dimensional subspace $\cV_1$ spanned by $|\chi_1\>$ is invariant under $U$ and the eigenvector of $U$ in $\cV_1$ is 
\begin{equation}
    |\psi_1\> = |\chi_1\> \in \mathcal{B}
\end{equation}
with eigenvalue $1$.

\item For $j\ge 2$, the two-dimensional subspace $\cV_j$ spanned by $|\chi_j\>$ and $S|\chi_j\>$ is invariant under $U$ and the two eigenvectors of $U$ in $\cV_j$ are
\begin{equation}
    |\psi^{\pm}_j\> = |\chi_j\> - \mu^{\pm}_j S|\chi_j\> \in \mathcal{B}
\end{equation}
with corresponding eigenvalues $\mu^{\pm}_j$, where
\begin{equation}
    \mu^{\pm}_j = \lambda_j \pm i \sqrt{1-\lambda_j^2} = e^{\pm i \, \mathrm{arccos}(\lambda_j)}.
\end{equation}
\end{enumerate}
\end{theorem}

\begin{proof}
We use the properties of the isometry $T$ to proceed.  We begin with the case $j\ge 2$, that is, $\lambda_j$ is bounded away from $1$ by at least the spectral gap $\Delta$. First, we obtain 
\begin{align}
    U |\chi_j\> 
    &= S (2\Pi - I)|\chi_j\> \\
    &= S (T T^\dagger - I) T |\varphi_j\> \\
    &= 2 S T |\varphi_j\> - ST|\varphi_j\> \\
    &= S |\chi_j\>.
\end{align}
Second, we obtain
\begin{align}
    U S|\chi_j\> 
    &= S(2\Pi - I) S T |\varphi_j\> \\
    &= S(2 T T^\dagger - I) ST |\varphi_j\> \\
    &= (2S T T^\dagger S T - T) |\varphi_j\> \\
    &= (2S T D - T) |\varphi_j\> \\
    &= (2\lambda_j S T - T) |\varphi_j\> \\
    &= (2\lambda_j S - I) |\chi_j\>.
\end{align}
We see that the subspace spanned by $|\chi_j\>$ and $S|\chi_j\>$ is invariant under $U$ so we can find eigenvectors of $U$ within this subspace. Let
\begin{equation}
    |\psi^{\pm}_j\> = |\chi_j\> - \mu^{\pm}_j S|\chi_j\>.
\end{equation}
We have
\begin{align}
    U|\psi^\pm_j\> 
    &= S|\chi_j\> - \mu^\pm_j (2\lambda_j S - I)|\chi_j\> \\
    &= \mu^\pm_j |\chi_j\>  - (2\lambda_j \mu^\pm_j - 1) S |\chi_j\>
\end{align}
Therefore, $|\psi^\pm_j\>$ is an eigenvector of $U$ with eigenvalue $\mu^\pm_j$ provided that 
\begin{equation}
    (\mu^\pm_j)^2 - 2 \lambda_j\mu^\pm_j + 1 = 0,
\end{equation}
that is,
\begin{equation}
    \mu^\pm_j = \lambda_j \pm i \sqrt{1 - \lambda_j^2} = e^{\pm i \, \mathrm{arccos}(\lambda_j)}.
\end{equation}

We now consider the case $j=1$.  We have
\begin{equation}
    1 = \<\varphi_1|Q|\varphi_1\> = \<\varphi_1|T^\dagger S T|\varphi_1\> = \<\chi_1|S|\chi_1\>,
\end{equation}
implying that $|\chi_1\>$ is an eigenvector of $S$ with eigenvalue $1$.  We have
\begin{equation}
    U|\chi_1\> = S(2\Pi - I)|\chi_1\> =  S |\chi_1\> = 1,
\end{equation}
where we used $\Pi|\chi\>=T T^\dagger T|\varphi_1\>= T|\varphi_1\>=|\chi_1\>$.
\end{proof}


\begin{lemma}[Lower bound on quantum gap]
Let $\Delta$ denote the classical gap of $Q$. The quantum gap of the corresponding unitary $U$ is $\theta=\arccos(1-\Delta)$ and is quadratically larger than the classical gap because 
\begin{equation}
    \theta \ge \sqrt{2\Delta}.
\end{equation}
\end{lemma}


\begin{proof}
Write the second largest eigenvalue as $\lambda_2 = 1 - \Delta = \cos(\theta)$ for a suitable $\theta$.  We have
\begin{equation}
    1 - \Delta = \cos(\theta) \ge 1 - \frac{\theta^2}{2},
\end{equation}
which implies the bound $\theta\ge\sqrt{2\Delta}$. \qed{}
\end{proof}


\end{document}